\documentclass[a4paper,11pt]{article}

\usepackage[utf8]{inputenc}
\usepackage[english]{babel}
\usepackage[a4paper,outer=2cm,inner=2cm,top=2cm,bottom=2cm]{geometry}
\usepackage{xcolor}  
\usepackage{amssymb,amsmath,amsthm}
\usepackage[hyperfootnotes=false, hypertexnames=false, pdfusetitle]{hyperref}
   
\hypersetup{colorlinks, allcolors={blue!65!black}}
 
\newtheorem{theorem}{Theorem}

\newtheorem{lemma}[theorem]{Lemma}
\newtheorem{corollary}[theorem]{Corollary}

\newtheorem{definition}[theorem]{Definition}
\newtheorem{observation}[theorem]{Observation}

\theoremstyle{plain}
\newtheorem*{rep@theorem}{\rep@title}
\newcommand{\newreptheorem}[2]{%
	\newenvironment{rep#1}[1]{%
		\def\rep@title{#2 \ref{##1} (restated)}%
		\begin{rep@theorem}}%
		{\end{rep@theorem}}}
\newreptheorem{theorem}{Theorem}

\let\epsilon=\varepsilon
 
\def\Omegainit{\Omega_{0}}
\def\calD{\mathcal{D}}
\newcommand{\taumin}[1]{{\tau_{#1^-}}}  
\newcommand{\taumax}[1]{{\tau_{#1^+}}}
\newcommand{\tauord}{\tau_{\mathrm{ord}}}
\def\E{\mathbb{E}} 
\def\Exp{\mathrm{Exp}} 
\newcommand{\dt}{\hat{\Delta}} 
\newcommand{\tildeM}{\tilde{M}}
\def\tildef{\tilde{f}} 
\def\tildeF{\tilde{F}}
\def\tildeV{\tilde{V}}
\def\tildepi{\tilde{\pi}}
\def\mut{\mathrm{mut}}
\def\dmut{D_{\mut}} 
\def\Dmut{\dmut}
\newcommand{\maxfixprob}{\textsc{MaxFixProb}}
\newcommand{\maxfixprobmut}{\textsc{MaxFixProb}_{\mut}}
\def\vtuple{\mathbf{u}}
\def\ttuple{\boldsymbol \gamma}
\def\nonnegints{\mathbb{Z}_{\geq 0}}
\def\nonnegreals{\mathbb{R}_{\geq 0}}
\newcommand{\absorb}{A}

\title{Parameterised Approximation of the Fixation Probability of the Dominant Mutation in the Multi-Type Moran Process\footnote{This research was funded in whole, or in part, by the German Academic Scholarship Foundation. For the purpose of Open Access, the authors have applied a CC BY public copyright licence to any Author Accepted Manuscript version arising from this submission. All data is provided in full in the results section of this paper.}}
\author{Leslie Ann Goldberg \\ Department of Computer Science\\ University of Oxford\\ United Kingdom
\and
Marc Roth \\ Department of Computer Science \\ University of Oxford\\ United Kingdom 
\and
Tassilo Constantin Schwarz \\ Department of Applied Mathematics and Theoretical Physics \\ University of Cambridge\\ United Kingdom 
}
\date{14 March, 2023}

\begin{document}

\maketitle

\begin{abstract}
The multi-type Moran process is an evolutionary process on a connected graph $G$ in which each vertex has one of $k$ types and, in each step, a vertex $v$ is chosen to reproduce its type to one of its neighbours. The probability of a vertex $v$ being chosen for reproduction is proportional to the fitness of the type of $v$.
So far, the literature was almost solely concerned with the $2$-type Moran process in which each vertex is either healthy (type $0$) or a mutant (type $1$), and the main problem of interest has been the (approximate) computation of the so-called \emph{fixation probability}, i.e., the probability that eventually all vertices are mutants.

In this work we initiate the study of approximating fixation probabilities in the multi-type Moran process on general graphs. Our main result is an FPTRAS (fixed-parameter tractable randomised approximation scheme) for computing the fixation probability of the dominant mutation; the parameter is the number of types and their fitnesses.
In the course of our studies we also provide novel upper bounds on the expected \emph{absorption time}, i.e., the time that it takes the multi-type Moran process to reach a state in which each vertex has the same type.
\end{abstract}

\section{Introduction}
The study of the ($2$-type) Moran process dates back to the late 1950s~\cite{Moran58} with the goal of modelling and understanding how an advantageous mutation spreads through a finite community: In its initial form, the process starts with $n$ individuals, $n-1$ of which are healthy (fitness $1$), and one of which is a ``mutant'' with an advantage (fitness $r>1$). In each step of the process, an individual is chosen with probability proportional to its fitness and replaces the type of another individual, which is chosen uniformly at random, with its own type (healthy or mutant). This process becomes stable if either all individuals are mutants, or if the mutation is extinct. It is well-known that the probability that the mutation takes over, which is called the \emph{fixation probability}, is $(1-r)/(1-r^n)$ (see e.g.\ Lieberman, Hauert, and Nowak~\cite{lieberman-EvolutionaryDynamicsGraphs-2005}).

The Moran process was later generalised to communities with a spatial structure~\cite{lieberman-EvolutionaryDynamicsGraphs-2005,Nowak06}. In the generalisation, individuals are vertices of an $n$-vertex graph $G$ and each vertex can only replace the type of its neighbours. Moreover, the single mutant at the start of the process is chosen uniformly at random from all vertices. For this \emph{spatial} Moran process, the fixation probability crucially depends on the structure of the graph. This quantity can be calculated explicitly, but the number of equations for doing so would grow exponentially in the graph’s size. Hence, D{\'i}az, Goldberg, Mertzios, Richerby, Serna and Spirakis~\cite{diaz-ApproximatingFixationProbabilities-2014} constructed a \emph{fully polynomial randomised approximation scheme (FPRAS)} for the associated computational problem.
This means that they gave an algorithm that takes as input a graph $G$, a mutation fitness $r\geq 1$, and real numbers~$\epsilon$ and $\delta$ in the interval $(0,1)$ that guide the the accuracy of the approximation. With probability at least~$1-\delta$,
the algorithm outputs an $\varepsilon$-approximation of the fixation probability, which means a number which is within $1\pm \epsilon$ of the fixation probability. According to the definition of FPRAS, the running time of this algorithm is bounded from above by a polynomial in $|G|$, $1/\varepsilon$ and $\log(1/\delta)$. For more related work on the $2$-type Moran process including improved FPRASes, see Section~\ref{sec:related_work}.

Much less is known about the Moran process 
when there is more than one type of  mutation. Although the multi-type Moran process has been considered for the special case of complete graphs by Ferreira and Neves~\cite{ferreira-FixationProbabilitiesMoran-2020}, there is, to the best of our knowledge, no work on the multi-type Moran process on general graphs.

In the setting of multiple types, mutated individuals will compete not only against healthy individuals, but different mutation lineages will also fight for dominance against each other. For example, this is the case in the formation of cancer in which typically a tumour contains two to eight ``driver'' mutations and hundreds of ``passenger'' mutations (see e.g.\ Vogelstein, Papadopoulos, Velculescu, Zhou, D{\'i}az, and Kinzler~\cite{vogelstein-CancerGenomeLandscapes-2013}).
For this reason, we initiate in this work the study of the multi-type Moran process on general graphs. The goal is to compute the fixation probability of the dominant mutant, i.e., the type with maximum fitness: If this probability is small enough, then we can hope for the healthy individuals/cells to eventually take over.  

Our main result is an efficient algorithm for approximating the fixation probability of the dominant mutant.
For the formal statement of our results, we will next introduce the process in detail.

\subsection{The Multi-Type Moran Process}

The graphs that we consider in this work are simple, undirected,  and without loops. Given a vertex $v$ of a graph $G=(V,E)$, we write $N(v):= \{u \mid\{u,v\}\in E\}$ for the (open) neighbourhood of $v$, and we write $d(v):=|N(v)|$ for the degree of $v$.

Let $G=(V,E)$ be an $n$-vertex graph. Let $\tau$ be a finite set of types 
and let $\tauord$ be an element of~$\tau$. 
Intuitively, $\tauord$ is the ``ordinary'' type --- all other types in~$\tau$ are mutations.
Let~$f$ be a function from~$\tau$ to $\mathbb{Q}_{\geq 1}$. 
We refer to~$f$ as a ``fitness function'' because we use it to
measure the ``fitness'' of the types in~$\tau$.
Let $f^+ = \max_{i \in \tau} f(i)$ and $\taumax{f} = \{ i \in \tau \mid f(i) = f^+\}$.
Similarly, let
$f^- = \min_{i\in \tau} f(i)$ and $\taumin{f} = \{i \in \tau \mid f(i) = f^-\}$.
In this work, we will 
study the case where mutations are advantageous, so $\tauord \in \taumin{f}$.
We say that $f$ is an \emph{advantageous fitness function} for $\tau$ and $\tauord$ in this case.

Let $\Omega$ be the set of functions from~$V$ to~$\tau$.
The elements of $\Omega$ will be the states that we study --- each state assigns a type to every vertex of~$G$.
Given a state $M_0 \in \Omega$, the
\emph{Moran process} 
corresponding to $G$, $\tau$, $f$ and $M_0$ is  a Markov Chain $M=M_0,M_1,\ldots$ with state space $\Omega$
and start state~$M_0$.

Before defining the transitions of~$M$, we give some notation. If $S$ is a state in $\Omega$
  and $v$ and $w$ are vertices in $V$ then $S\vert_{v\rightarrow w}$ is the state $S'$
such that $S'(w) = S(v)$ and, for all $u\in V\setminus \{w\}$, $S'(u) = S(u)$.
Informally, $S'$ is derived from $S$ by reproducing the type of $v$ onto the type of $w$. For example, if $v$ is a mutation type, then this mutation spreads to its neighbour~$w$ in $S\vert_{v\rightarrow w}$.

We refer to $M_t$ as the state of the Markov Chain~$M$ at (discrete) time~$t$.
For any type $j\in\tau$, we write $V_j(t)$ for the set of vertices that are assigned type $j$ at time $t$, that is, $V_j(t):=\{v \in V \mid M_t(v) = j\}$.  To simplify the notation, we set $f(v,t):=f(M_t(v))$ for the fitness of 
the type of~$v$ at time~$t$. Given a subset $U\subseteq V$, we write
$ f(U,t):=\sum_{v\in U}f(v,t)$ to denote
the total fitness of the vertices in $U$ at time~$t$. The \emph{total fitness} $F(t)$ at time $t$ is the total fitness of all vertices, that is, $F(t):=f(V,t)$.

The transition from $M_t$ to $M_{t+1}$ is defined as follows.
A random vertex $v\in V$ is chosen with probability $f(v,t)/F(t)$.
A vertex $w\in N(v)$ is chosen with probability $1/d(v)$.
Then $M_{t+1} := M_t\vert_{v\rightarrow w}$.
Intuitively, in each transition,   a vertex $v$ is chosen with probability proportional to the fitness of its type at time $t$ and a vertex $w$ is chosen uniformly at random from the neighbours of $v$. Then the type of~$w$ is replaced with the type of~$v$.

If $M_t$ assigns the same type to all vertices then the state of the Markov chain~$M$ does not change after time~$t$. 
We say that type $j\in \tau$ \emph{fixates} by time~$t$ if 
$M_t$ assigns type $j$ to all vertices. 
If $G$ is connected, then with probability~$1$, some type in~$\tau$ fixates.
We use $\pi_j$ to denote the probability that $j$ fixates.

If we want to clarify $G$, $\tau$, $f$, and $M_0$, 
then we refer to $M$ as $M(G,\tau,f,M_0)$ and we refer to $\pi_j$ as
$\pi_j(G,\tau,f,M_0)$.

Sometimes it will be useful to consider a distribution on the start state~$M_0$.
If $D$ is a distribution on $\Omega$ then
$\pi_j(G,\tau,f,D) = \sum_{M_0\in \Omega} \Pr_D(M_0) \, \pi_j(G,\tau,f,M_0)$.

For our purposes, not every distribution~$D$ on~$\Omega$ will be important. 
In the well-studied case where 
$|\tau|=2$ and the fitness function is advantageous (so  
$\tauord \in \taumin{f}$)  
the idea is that the type $\alpha \in \tau \setminus \{\tauord\}$ 
is a ``mutant'' type with a high fitness that arises as a result
of a single mutation. It is therefore appropriate 
 to consider
the 
case where the start state~$M_0$ is drawn from the  
distribution~$\dmut=\dmut(G,\tau,\tauord)$
which is uniform on states $M_0\in \Omega$ in which
exactly one vertex has the type in $\tau \setminus \{\tauord\}$.
This is the distribution that has been studied in previous works.

There is a  natural generalisation of~$\dmut$ to the multi-type Moran process
where $|\tau|=k>2$.  In this case, $\dmut$ is the uniform distribution on states $M_0 \in \Omega$ which have exactly one vertex assigned to every type $i\in \tau \setminus \{\tauord\}$.

Our results will apply to $\dmut$ and also to several other distributions. 
Let $V[k]$ be the set containing all $k$-tuples of distinct vertices
from~$V$  and 
let $\tau[k]$ be the set containing all $k$-tuples of distinct types in $\tau$.
For every $\vtuple\in V[k]$ and $\ttuple \in \tau[k]$, 
let $\Omega(\vtuple,\ttuple)$ be the set of states in $\Omega$ 
that map vertices in $\vtuple$ respectively to types in $\ttuple$.
When a state $S\in \Omega$ is drawn from~$\Dmut$, there is a clear lower bound on the probability that $S\in \Omega(\vtuple,\ttuple)$. Our results will apply to any distribution with this property. 
For concreteness, we identify a set of suitable distributions as follows.
The set of distributions $\calD = \calD(G,\tau)$ is
the set containing every distribution $D$ on $\Omega$  
that meets the following criteria.
\begin{itemize}
\item For every pair $(\vtuple,\ttuple) \in V[k] \times \tau[k]$,
$D_{\vtuple,\ttuple}$ is a distribution on $\Omega(\vtuple,\ttuple)$.
\item
For any $S\in \Omega$, 
$ \Pr_D(S) = \frac{1}{\strut |V[k] \times \tau[k]|} 
\sum_{(\vtuple,\ttuple) \in V[k] \times \tau[k]} \Pr_{D_{\vtuple,\ttuple}} (S)$.
\end{itemize}

It is easy to see that $\dmut(G,\tau,\tauord) \in \calD(G,\tau)$ since, in this case, $D_{\vtuple,\ttuple}$ can be taken to be the distribution containing the single state in 
$\Omega(\vtuple,\ttuple)$
such that every vertex outside of $\vtuple$ is assigned type $\tauord$.

Let $\Omegainit(G,\tau)$ be the set of all states $S\in \Omega$ 
with range~$\tau$. By construction, any sample $S$ drawn from a distribution $D\in \calD(G,\tau)$ is in $\Omegainit(G,\tau)$.

\subsection{Parameterised and Approximation Algorithms}

We will study the following problems concerned with the fixation probability of a type with maximum fitness.
\begin{definition}[$\maxfixprobmut$]
The problem $\maxfixprobmut$ takes as input a connected graph $G$, a set of types $\tau$ with ordinary type~$\tauord\in \tau$, an advantageous fitness function $f\colon \tau \to \mathbb{Q}_{\geq 1}$, and a type $\alpha\in \taumax{f}$.
The goal is to compute $\pi_{\alpha}(G,\tau,f,\dmut(G,\tau,\tauord))$.
\end{definition}

\begin{definition}[$\maxfixprob$]
The problem $\maxfixprob$ takes as input a connected graph $G$, a set of types $\tau$, a fitness function $f\colon \tau \to \mathbb{Q}_{\geq 1}$, a type $\alpha\in \taumax{f}$ and a distribution $D\in \calD(G,\tau)$ (via a black-box oracle which provides samples from~$D$).
The goal is to compute $\pi_{\alpha}(G,\tau,f,D)$.
\end{definition}

As explained in the Introduction, we are interested in efficient approximation algorithms for $\maxfixprob$ and $\maxfixprobmut$.

As is usual in the setting of parameterised algorithms
(and is natural in the case of this application),
we will assume that the number of types (and the encoding of their fitnesses) is significantly smaller than the size of the
input graph~$G$, which we denote by~$|G|$. Thus, we are happy to accept a factor 
in the running time that depends on $\tau$ and $f$, say $g(\tau,f)$ --- but we would not be happy with an algorithm whose running time is as high as $|G|^{g(\tau,f)}$ for any function~$g$.

Formally, we let $||f||$ be the description length of the fitness function~$f$. The exact details about how~$f$ is encoded will not be important here. We \emph{parameterise} our problems in terms of the parameter $\kappa := |\tau|+||f||$ and we aim to construct a \emph{fixed-parameter tractable approximation scheme (FPTRAS)}, which is the standard notion for efficient parameterised approximation algorithms
(see Arvind and Raman~\cite{ArvindR02}).
\begin{definition}[FPTRAS]
A \emph{fixed-parameter tractable randomised approximation scheme} (FPTRAS) for a problem $P$ is a randomised algorithm $\mathbb{A}$ that takes as input
a problem
input $x$ and reals $0<\varepsilon,\delta < 1$, and outputs $\hat{P}$ such that
$\Pr((1-\varepsilon) P(x) \leq \hat{P} \leq (1 + \varepsilon) P(x)) \geq 1 - \delta$.
Moreover, there must be a computable function $g$ such that the running time of $\mathbb{A}$ is bounded from above by
$g(\kappa) \cdot \mathsf{poly}(|x|,1/\varepsilon, \log (1/\delta))$, where $\kappa$ is the parameter of $x$.
\end{definition}

\subsection{Main Results}

We are now able to state our main results.

\newcommand{\stateThmMain}{ $\maxfixprob$ has an FPTRAS when parameterised by $\kappa := |\tau|+||f||$.}

\begin{theorem}\label{thm:main}
   \stateThmMain
\end{theorem}

\newcommand{\stateThmMaxfixMutHasFptras}{$\maxfixprobmut$ has an FPTRAS when parameterised by $\kappa := |\tau|+||f||$.}
\begin{theorem} \label{thm:maxfix-mut-has-fptras}
   \stateThmMaxfixMutHasFptras
\end{theorem}

In the course of proving Theorems~\ref{thm:main} and~\ref{thm:maxfix-mut-has-fptras} we also obtain upper bounds on the absorption time, i.e. the time that the process runs until a type fixates.
For the 
multi-type Moran process, this upper bound is of the form $g(|\tau|+||f||)\cdot \mathsf{poly}(|G|)$ (see Theorem~\ref{thm:absorption_time}). 

\subsection{Related Work}\label{sec:related_work}

\subsubsection*{Approximation algorithms for the 2-type spatial Moran process}

Work on approximating the fixation probability in the spatial Moran process has so far focused on the 2-type process, for which D{\'i}az et al.\ \cite{diaz-ApproximatingFixationProbabilities-2014} developed the first FPRAS. Denoting the types by $\tau= \{ 1,2\}$ with $\tauord = 1$ being the ordinary type, the idea to approximate the fixation probability $\pi_2$ of the mutation is to perform sufficiently many Markov chain Monte Carlo simulations of the process, where the number of steps and simulations can be bounded by proving (a) an upper bound on the expected absorption time and (b) a lower bound on the fixation probability $\pi_2$. 
Their work applies to the case where the mutation is advantageous, so $\taumin{f} = \{1\}$ and $f(2) > f(1)$. Note that when the mutation is disadvantageous, its fixation probability $\pi_2$ can be exponentially small as a function of the size of the graph, even when the graph is the complete graph, so this approach to approximation is infeasible.

In a first improvement, Chatterjee, Ibsen-Jensen, and Nowak~\cite{chatterjee-FasterMontecarloAlgorithms-2017} exploit the observation that a vertex reproducing to a neighbour of the same type does not change the overall state. By simulating only those steps that change the underlying Markov chain's state, they speed the FPRAS up.

The effect of $f(2)/f(1)$ on the fixation probability $\pi_2$ in \emph{every} connected, undirected graph was characterised by Goldberg, Lapinskas, and Richerby~\cite{anngoldberg-PhaseTransitionsMoran-2020} by demonstrating a phase transition in $\pi_2$ as a function of $f(2)/f(1)$. This phase transition also entails an improved lower bound on~$\pi_2$ in the advantageous case.  Combining this bound with early termination of the Monte Carlo simulation once fixation is sufficiently likely, they provide a faster FPRAS.

\subsubsection*{Multiple types in related models}

The idea of modelling multiple mutation types has been explored for both the Wright-Fisher model~\cite{beerenwinkel-GeneticProgressionWaiting-2007} and, more recently, the non-spatial Moran process~\cite{etheridge-CoalescentDualProcess-2009,ferreira-FixationProbabilitiesMoran-2020,chao-SpatialMutationModel-2021}. 
The Wright-Fisher model~\cite{fisher-DominanceRatio-1923,wright-EvolutionMendelianPopulations-1931} is a non-spatial model with a coarse-grained time scale where all individuals reproduce simultaneously, resulting in non-overlapping generations. This model served as the predecessor of the more fine-grained Moran process; see Lanchier~\cite{lanchier-WrightFisherMoran-2017} for a detailed definition and comparison to the non-spatial Moran process.

For the Wright-Fisher model, Beerenwinkel, Antal, Dingli, Traulsen, Kinzler, Velculescu, Vogelstein and Nowak~\cite{beerenwinkel-GeneticProgressionWaiting-2007} provide a generalisation to account for the evolution of multiple types: They consider a sequentially ordered  set $\tau$ of 
infinitely many types, where higher types have higher fitness. Initially, all individuals are of some ordinary type $\tauord$. Besides reproduction happening generation by generation, individuals can progress to a higher type at some rate
which depends on how close the initial type is to the new one in the ordering of~$\tau$.
Beerenwinkel et al.\ perform computer simulations of this process which models the evolution of cancer cells. Their goal is to approximate the expected 
time until a cancer phenotype appear, which they define as the  waiting time 
until some fixed type occurs as a mutation. The focus of their model and simulations is thus on the progression of types.

Etheridge and Griffiths~\cite{etheridge-CoalescentDualProcess-2009} define for the non-spatial Moran process a finite set of types $\tau$ where, in addition to the reproduction dynamics, 
mutation happens between different types at given rates. Chao and Schweinsberg~\cite{chao-SpatialMutationModel-2021} extend the study to a population on a continuous torus, where the focus is on the progression between types  (similar to Beerenwinkel et al.~\cite{beerenwinkel-GeneticProgressionWaiting-2007}).

For the Moran process without spatial dynamics and without progression between types, Ferreira and Neves~\cite{ferreira-FixationProbabilitiesMoran-2020} introduce in recent work the concept of multiple types. 
Their model is the same as ours except that they focus on the special case where the graph $G$ is the complete graph.
They provide bounds on the fixation probabilities in the 3-type case. In the context of our work on the \emph{spatial} Moran process, these bounds can be derived by applying our coupling in the special case 
where $G$ is a complete graph and   $|\tau| = 3$,   as we show in Lemma~\ref{lem:complete}.

\section{The Continuous Moran Process}  

Although our main results are about the discrete-time Moran process, it is useful in proofs to  consider a continuous-time version of the Moran process.

We use $\Exp(\lambda)$ to denote an exponential distribution with parameter~$\lambda$.
Let $M$ be the (discrete-time) Moran process 
corresponding to $G=(V,E)$, $\tau$, $f$, and $M_0$. 
Following the work of Diaz, Goldberg, Richerby, and Serna~\cite{diaz-AbsorptionTimeMoran-2016} 
in the case where $|\tau|=2$, we define
a continuous-time Moran process~$\tildeM = \tildeM(G,\tau,f,M_0)$ corresponding to~$M$.  
The process $\tildeM$ has the same state space as $M$,
namely the set of function~$\Omega$ from~$V$ to~$\tau$.
and it starts with $\tildeM_0 = M_0$. 
The evolution of $\tildeM$ is guided by a sequence $T_0,T_1,\ldots$ of non-negative real numbers with $T_0=0$.
For any non-negative integer~$i$, the
evolution from $\tildeM_{T_i}$ proceeds as follows.
Each vertex $v\in V$ is equipped with a ``clock'' $c^i_{v}$ which is an exponentially distributed random variable  with parameter~$\tilde{f}(v,T_i) := 
f(\tildeM_{T_i}(v))$. The distributions of the clocks $c^i_v$
are mutually independent.  Let $\Delta_i := \min_{w \in V} c^i_w$ and $T_{i+1} = T_i + \Delta_i$.
For $t \in [T_i,T_{i+1})$,  we define $\tildeM_t = \tildeM_{T_i}$.
We define $\tildeM_{T_{i+1}}$  
as follows. Let $v := \arg \min \{c^i_u \mid u\in V\}$.   
Choose $w\in N(v)$ with probability $1/d(v)$. 
Then 
$\tildeM_{T_{i+1}} := \tildeM_{T_i}\vert_{v\rightarrow w}$.

The reason that D{\'i}az et al.~\cite{diaz-AbsorptionTimeMoran-2016} chose the exponential distribution for the clocks
is that it is memoryless, so it is convenient for couplings.
Moreover, the minimum of a set $S$ of exponential distributions is 
exponentially distributed with a new parameter that is  the sum of the parameters in $S$. Therefore, for all $v$ and $T_i$,  
 $
    \Pr(\Delta_i = c^i_v) =  {\tilde{f}(v,T_i)}/{\tildeF(T_i)}  
 $ where $\tildeF(T_i) = \sum_{v\in V} \tildef(v,T_i)$. This is the same as the probability that $v$ is chosen for reproduction
 in a discrete step starting from $\tildeM_{T_i}$.
 Thus, the process $\tildeM_{T_0},\tildeM_{T_1},\ldots$ is a faithful copy of $M$.
It follows that $\tilde\pi_j(G,\tau,f,\tildeM_0)$, the probability that $j$~fixates in $\tildeM$, is equal to $\pi_j(G,\tau,f,M_0)$.

\section{Coupling Continuous Multi-Type Moran Processes}

 The following lemma is a generalisation of \cite[Lemma 5]{diaz-AbsorptionTimeMoran-2016}, which was for the case $k=2$.

\begin{lemma} \label{lem:fixation:coupling}\label{lem:coupling}
Let $G=(V,E)$ be a graph, let $\tau$ be a set of types, and let $f,f'\colon\tau \to \mathbb{Q}_{\geq 1}$ be   fitness functions. Suppose that 
there is a type $\alpha \in  \tau$ with  $f(\alpha)=f'(\alpha)$ and that, for any types $\beta,\beta' \in \tau \setminus \{\alpha\}$,
$f(\beta) \leq f'(\beta')$.
Let $\tildeM_0$ and $\tildeM'_0$ be states in the set $\Omega$ of functions from $V$ to~$\tau$.
Let $\tildeM = \tildeM(G,\tau,f,\tildeM_0)$ and $\tildeM' = \tildeM'(G,\tau,f',\tildeM'_0)$ be continuous-time Moran processes. For all $j\in \tau$, let $\tildeV_j(t)=\{v \in V \mid \tildeM_t(v) = j\}$ and 
$\tildeV'_j(t)=\{v \in V \mid \tildeM'_t(v) = j\}$.
Suppose that $V_{\alpha}'(0)\subseteq V_{\alpha}(0)$.
Then
 there is a coupling between $\tildeM$ and $\tildeM'$ such that, for any  time $t \in \nonnegreals $, $V_{\alpha}'(t)\subseteq V_{\alpha}(t)$.
	\end{lemma}

\begin{proof}

For convenience, let  
$\tildef(v,T_i) := f(\tildeM_{T_i}(v))$ and 
$\tildef'(v,T_i) := f(\tildeM'_{T_i}(v))$.  
The joint process is guided by a sequence $T_0,T_1,\dots$ of non-negative real numbers with $T_0 = 0$. At time $T_0$, both processes are in their respective initial states $\tildeM_0$ and $\tildeM'_0$.  For any non-negative integer $i$, the evolution of the joint process from time $T_i$ is defined as follows. First, we define some exponentially-distributed random variables.
\begin{align*}
        \forall v \in V: \quad C^0_v(i) &\sim \Exp(\min\left(\tildef(v,T_i), \tildef'(v,T_i)\right)).\\
        \forall v \in V \text{ with } \tildef(v,T_i) > \tildef'(v,T_i):\quad  C_v^1(i) &\sim \Exp(\tildef(v,T_i) - \tildef'(v,T_i)).\\
        \forall v \in V \text{ with } \tildef(v,T_i) < \tildef'(v,T_i):\quad  C_v^2(i) &\sim \Exp(\tildef'(v,T_i)- \tildef(v,T_i)).
    \end{align*}

The earliest of these random variables is the random variable $\dt(i)$  defined by
 $$\dt(i) := \min   \left\{C^0_v(i) \mid  v \in V\right\} \cup \left\{C_v^1(i) \mid \tildef(v,T_i) > \tildef'(v,T_i) \right\}\cup \left\{C_v^2(i) \mid \tildef(v,T_i) < \tildef'(v,T_i)\right\}. $$

Let $T_{i+1} = T_i + \dt(i)$.   The two chains evolve from $T_i$  as follows. 
Both $\tildeM$ and $\tildeM'$ stay in their respective states over the interval $[T_i, T_{i+1})$.
At time~$T_{i+1}$ the state changes as follows.
\begin{description}
\item[Case 0.] If $\dt(i) = C^0_v(i)$ for some $v\in V$,
select $w \in N(v)$ u.a.r.{} and set $\tildeM_{T_{i+1}} := \tildeM_{T_{i}}\vert_{v \rightarrow w}$
and $\tildeM'_{T_{i+1}} := \tildeM'_{T_{i}}\vert_{ v \rightarrow w}$

\item[Case 1.] If $\dt(i) = C_v^1(i)$ for some $v\in V$,  select  $w \in N(v)$ u.a.r{} and set $\tildeM_{T_{i+1}} := \tildeM_{T_{i}}\vert_{v \rightarrow w}$
and $\tildeM'_{T_{i+1}} := \tildeM'_{T_{i}}$.

 \item[Case 2.] If $\dt(i) = C_v^2(i)$ for some $v \in V$,   select  $w \in N(v)$ u.a.r{} and 
  set $\tildeM_{T_{i+1}} := \tildeM_{T_{i}}$
and $\tildeM'_{T_{i+1}} := \tildeM'_{T_{i}}\vert_{ v \rightarrow w}$.

    \end{description} 
It is easy to see that this joint evolution is a correct evolution for the individual chains~$\tildeM$ and~$\tildeM'$. For example, consider $\tildeM$ (the argument for~$\tildeM'$ is symmetric). We will show that for every $v\in V$, the time until $v$ is chosen for reproduction is exponentially distributed with parameter $\tildef(v,T_i)$. 
If $\tildef(v,T_i) \leq \tildef'(v,T_i)$ this follows since $C^0_v(i)$
is exponentially distributed with parameter $\tildef(v,T_i)$ and $\tildeM$ does not evolve in Case~2.
Otherwise, it follows since  
    $\min \{C_v^0(i), C_v^1(i)\}$ 
is exponentially distributed with parameter $\tildef(v,T_i)$. 
Now assume  $V_{{\alpha}}'(T_i)\subseteq V_{{\alpha}}(T_i)$. To finish the proof, we will show $V_{{\alpha}}'(T_{i+1})\subseteq V_{{\alpha}}(T_{i+1})$. Specifically, we
will rule out the possibility  that  $w \in V_{\alpha}'(T_{i+1}) \setminus V_{\alpha}(T_{i+1})$.

In Case~0 the only way that $w\in V_{\alpha}'(T_{i+1})$ is that 
$v\in V_{\alpha}'(T_i)$.  This implies 
$v\in V_{\alpha}(T_i)$ so
$w\in V_{\alpha}(T_{i+1})$, as required.

In Case~1 the only way that  $w \in V_{\alpha}'(T_{i+1}) \setminus V_{\alpha}(T_{i+1})$
is that   
$v\notin V_{\alpha}(T_i)$ so $v\notin V_{\alpha}'(T_i)$. But then Case~1 requires
$\tildef(v,T_i) > \tildef'(v,T_i)$ and this is not possible since $v\notin V_{\alpha}(T_i)$ 
and $v\notin V_{\alpha}'(T_i)$
(given the condition about types $\beta,\beta'$ in the lemma statement).

In Case~2 the only way that $w\in V_{\alpha}'(T_{i+1})$ is that 
$v\in V_{\alpha}'(T_i)$ so $v\in V_{\alpha}(T_i)$.  But then we do not have $\tildef(v,T_i) < \tildef'(v,T_i)$ as required for Case~2. 
\end{proof}

Lemma~\ref{lem:coupling} implies Corollary~\ref{cor:lower_bound_probability} which gives a lower bound on 
the probability that a maximally-fit mutation fixates in distributions in~$\calD$.
In order to prove Corollary~\ref{cor:lower_bound_probability} we 
will use Corollary~\ref{cor:tech}, which reduces
the problem of lower-bounding the fixation probability in the multi-type Moran process to
the problem of lower-bounding the fixation probability in the two-type Moran process.
We will also use  
Lemma~\ref{lem:oldstuff}, which is
a consequence of earlier work on the two-type Moran process.

 \begin{corollary}\label{cor:tech}
Let $G=(V,E)$ be a connected graph, let $\tau$ be a set of types, and let $f\colon \tau \to \mathbb{Q}_{\geq 1}$ be a fitness function. Let $\alpha$ be a type in~$\tau$.
Let $\beta = \arg\max\{f(j) \mid j\in \tau\setminus\{\alpha\}\}$.
Let $\tau'=\{\alpha,\beta\}$.
Let $\Omega$ be the set of functions from $V$ to $\tau$ and let $\Omega'$ be the set of functions from $V$ to $\tau'$. Let $g\colon \Omega \to \Omega'$ be defined as follows.
For every state $S\in \Omega$, $g(S)$ is the state  
such that $g(S)(v) = \alpha$ if $S(v)=\alpha$ and $g(S)(v)=\beta$, otherwise.
Then, for every 
state $M_0 \in \Omega$, $\pi_\alpha(G,\tau,f,M_0)\geq \pi_\alpha(G,\tau',f,g(M_0))$.
 \end{corollary}
\begin{proof}
There are three steps to the argument. First let  $f'\colon \tau \to \mathbb{Q}_{\geq 1}$ 
be the fitness function with $f'(\alpha)=\alpha$ and, for any $\beta' \in \tau\setminus \{\alpha\}$,
$f'(\beta') = f(\beta)$.
Note that for any $M_0\in \Omega$,
$\pi_\alpha(G,\tau',f,g(M_0)) = \pi_\alpha(G,\tau,f',M_0)$, so the goal is to
prove $\pi_\alpha(G,\tau,f,M_0)\geq \pi_\alpha(G,\tau,f',M_0)$.

Second, let $\tildeM = \tildeM(G,\tau,f,M_0)$ and let $\tildeM' = \tildeM'(G,\tau,f',M_0)$.
For all $j\in \tau$, let $\tildeV_j(t)=\{v \in V \mid \tildeM_t(v) = j\}$ and 
$\tildeV'_j(t)=\{v \in V \mid \tildeM'_t(v) = j\}$.
Applying Lemma~\ref{lem:coupling} 
we find that 
for any  time $t \in \nonnegreals$, $\tildeV_{\alpha}'(t)\subseteq \tildeV_{\alpha}(t)$.
so 
$\tildepi_\alpha(G,\tau,f,M_0) \geq \tildepi_\alpha(G,\tau,f',M_0)$.

Finally, recall
(from the construction of the continuous Moran process) that 
$\tildepi_\alpha(G,\tau,f,M_0) = \pi_\alpha(G,\tau,f,M_0)$
and
$\tildepi_\alpha(G,\tau,f',M_0) = \pi_\alpha(G,\tau,f',M_0)$.
Putting these together proves    the corollary.   
\end{proof}

\begin{lemma} \label{lem:oldstuff} 
Let $G=(V,E)$ be a connected graph, let $\tau$ be a set of types with $|\tau|=2$, and let $f\colon \tau \to \mathbb{Q}_{\geq 1}$ be a fitness function. Let $\alpha$
be a type in $\taumax{f}$.  Let $D$ be a distribution in $\calD(G,\tau)$.  Then 
$\pi_{\alpha}(G,\tau,f,D) \geq 1/n$.
\end{lemma}
\begin{proof}

Let $\tauord$ be the type in $\tau \setminus \{\alpha\}$. 
Recall the distribution $\Dmut = \Dmut(G,\tau,\tauord)$ which is uniform on states $M_0 \in \Omega$ in which exactly one vertex has type~$\alpha$.

We will  construct a coupling of~$D$ and~$\Dmut$ such that, in the random sample $(M_0,M^{\mut}_0)$ drawn from the coupling, every state that is assigned type~$\alpha$ in~$M^{\mut}_0$ is also assigned type~$\alpha$ in~$M_0$.
 We can then apply
Theorem~6 of \cite{diaz-AbsorptionTimeMoran-2016} to show that 
$\pi_{\alpha}(G,f,\tau,D) \geq \pi_{\alpha}(G,f,\tau,\Dmut)$, completing the proof.

To construct the coupling, consider the distribution $D$. By definition, 
$D$ is defined as follows: Choose $(\vtuple, \ttuple)$ uniformly at random from $V[2] \times \tau[2]$ and choose $M_0$ from  the distribution~$D_{\vtuple,\ttuple}$  on $\Omega(\vtuple,\ttuple)$. Let $n=|V|$. 
In the coupling, the outcome in which $M_0^{\mut}$ assigns vertex $u$ to type $\alpha$
is coupled with all choices  $(\vtuple,\ttuple)$ such that one of the following occurs:
(1) $u$ is the first vertex of $\vtuple$ and $\alpha$ is the first type in $\ttuple$ and
(2) $u$ is the second vertex of $\vtuple$ and $\alpha$ is the second type in $\ttuple$.
Note that $1/n$ of the pairs $(\vtuple, \ttuple)$ have this property. Also, every outcome in $D_{\vtuple, \ttuple}$ assigns $\alpha$ to $u$, as required.
\end{proof}

\begin{corollary}\label{cor:lower_bound_probability}\label{cor:low}
Let $G=(V,E)$ be a connected graph, let $\tau$ be a set of types, and let $f\colon \tau \to \mathbb{Q}_{\geq 1}$ be a fitness function. Let $\alpha$
be a type in $\taumax{f}$.  Let $D$ be a distribution in $\calD(G,\tau)$.  Then 
$\pi_{\alpha}(G,f,\tau,D) \geq 1/n$. \end{corollary}

\begin{proof}
Let $k=|\tau|$.
If $k=2$, the result follows directly from Lemma~\ref{lem:oldstuff}. Otherwise,
Let $\beta = \arg\max\{f(j) \mid j\in \tau\setminus\{\alpha\}\}$.
Let $\tau'=\{\alpha,\beta\}$.
Let $\Omega$ be the set of functions from $V$ to $\tau$ and let $\Omega'$ be the set of functions from $V$ to $\tau'$. Let $g\colon \Omega \to \Omega'$ be defined as follows.
For every state $S\in \Omega$, $g(S)$ is the state  
such that $g(S)(v) = \alpha$ if $S(v)=\alpha$ and $g(S)(v)=\beta$, otherwise.
Corollary~\ref{cor:tech} ensures that, 
for every 
state $M_0 \in \Omega$, $\pi_\alpha(G,\tau,f,M_0)\geq \pi_\alpha(G,\tau',f,g(M_0))$.
  
Let $g(D)$ be the distribution on~$\Omega'$ induced by~$D$ and~$g$. Specifically, to sample from $g(D)$, take a sample $S$ from~$D$ and then output~$g(S)$. 
Then $\pi_\alpha(G,\tau,f,D)\geq \pi_\alpha(G,\tau',f,g(D))$.
To finish, we
will show that $g(D) \in \calD(G,\tau')$.
This implies the corollary since then we can apply Lemma~\ref{lem:oldstuff} to conculde  
$\pi_\alpha(G,\tau',f,g(D)) \geq 1/n$.

To prove that $g(D) \in \calD(G,\tau')$,   recall what 
it means that $D\in \calD(G,\tau)$.
This means that
for every $\vtuple \in V[k]$ and $\ttuple \in \tau[k]$,
there is a distribution $D_{\vtuple,\ttuple}$  on $\Omega(\vtuple,\ttuple)$  such that
$D$ can be defined as follows: Choose $(\vtuple, \ttuple)$ uniformly at random from $V[k] \times \tau[k]$ and choose a sample $S$ from $D_{\vtuple,\ttuple}$.

For every pair $(\vtuple',\ttuple')$ in $V[2] \times \tau'[2]$
let $W_{\vtuple',\ttuple'}$ be  the set of pairs  $(\vtuple, \ttuple) \in  V[k] \times \tau[k]$ 
such that $\vtuple'$ is a sub-tuple of $\vtuple$
and $\ttuple'$ is the same sub-tuple (using the same indices) of $\ttuple$.
Let $D'_{\vtuple',\ttuple'}$ be the following distribution: 
Choose $(\vtuple, \ttuple)$ uniformly at random from~$W_{\vtuple',\ttuple'}$, Choose $S$ from $D_{\vtuple,\ttuple}$ and output $g(S)$. Let $D'$ be the distribution 
 defined as follows:  Choose $(\vtuple',\ttuple')$ u.a.r.{} from $V[2] \times \tau'[2]$
 and sample an output from $D'_{\vtuple',\ttuple'}$. 
It is clear from construction that $D'$   is in 
$\calD(G,\tau')$. So to prove  $g(D)\in \calD(G,\tau')$
we need only show $g(D)=D'$. 
This follows from that fact that, in the execution of $D'$, each pair $(\vtuple, \ttuple)\in V[k]\times \tau[k]$ is equally likely to be chosen. 
\end{proof}

\section{Bounding the Absorption Time}\label{sec:absorption}

The goal of this section is to bound the absorption time of the multi-type Moran process.  
It will be convenient to work with connected graphs in this section. We will also start from initial states 
$M_0\in \Omegainit(G,\tau)$ since all distributions in $\calD(G,\tau)$ are distributions on the state space~$\Omegainit(G,\tau)$.

\begin{definition}[Absorption time] \label{def:abtime}
Let $G=(V,E)$ be a connected graph, let $\tau$ be a set of types, let $f:\tau \to \mathbb{Q}_{\geq 1}$ be a fitness function and let $M_0\in \Omegainit(G,\tau)$ be an initial state.
For the Moran process $M=M(G, \tau, f,M_0)$ 
and for any type $j\in \tau$
we define the stopping time 
$$\absorb_j(G,\tau,f,M_0) := \min \{ 
\mbox{$t \in \nonnegints \mid V_j(t) = V $ or  $V_j(t) = \emptyset$} \}.$$  
We refer to~$\absorb_j$ as the \emph{absorption time} of type~$j$ in~$M$. The \emph{total absorption time} of~$M$ is defined as
$$\absorb(G,\tau, f,M_0) := \min \{ t \in \nonnegints \mid 
\mbox{ for some    
    $j\in\tau$, $V_j(t) = V$} 
    \}.$$
    If $G,\tau,f$, and $M_0$ are clear from the context, we will just write $\absorb_j$ and $\absorb$.
\end{definition}

We will upper bound the absorption time for each type, which will then easily translate to an upper bound on the total absorption time. To this end, we start with the most fit type.

\subsection{Expected absorption time of the most fit mutant} \label{sec:absorption:strongest}
To upper bound the expected absorption time, we adopt the potential function for the $2$-type process from the work of D{\'i}az et al.~\cite[Section 3]{diaz-ApproximatingFixationProbabilities-2014}
adapted for a maximally-fit type  in the multi-type process.
\begin{definition}\label{def:psi}
Let $G=(V,E)$ be a graph and let $\tau$ be a set of types.
Let $\Omega$ be the set of functions from~$V$ to~$\tau$.
For any type $j\in \tau$,
the potential function $\Psi_j$ is the function that maps every state $S\in \Omega$
to the value $\Psi_j(S) = \sum_{v: S(V) = j} 1/d(v)$.
\end{definition}

Lemma~\ref{lem:XX} applies in situations where the most fit type is unique. 
The proof is almost identical to the proof of~\cite[Lemma 5]{diaz-ApproximatingFixationProbabilities-2014}, but we include the details for completeness.

 \begin{lemma} \label{lem:potential-increasing-in-expectation_new}\label{lem:XX}
Let $G=(V,E)$ be a connected graph, let $\tau$ be a set of types with $|\tau|=k>1$, and let $f\colon \tau \to \mathbb{Q}_{\geq 1}$ be a fitness function. Let $\alpha$
be a type in $\taumax{f}$. 
Let $f^* =\max\{f(j) \mid j\in\tau\setminus\alpha\}$, and suppose that $f^*<f(\alpha)=f^+$. 
Suppose that   $V_\alpha(t)\notin \{\emptyset,V\}$ for a non-negative integer~$t$. Then    \begin{align*}  
        \E({\Psi_{\alpha}(M_{t+1}) - \Psi_{\alpha}(M_t)}) > (1-  {f^\ast}/{f^+} )/n^{3}.
    \end{align*}
\end{lemma}
\begin{proof}
First note that $\E(\Psi_{\alpha}(M_{t+1}) - \Psi_{\alpha}(M_t))$ is equal to
\begin{align*}
	~&\sum_{\substack{\{u,v\} \in E\\ u \in V_\alpha(t)\\ v \in V_j(t), j \neq \alpha}}  \frac{f^+}{F(t)} \frac{1}{d(u)}  \left(\Psi_{\alpha} ( M_t\vert_{u\rightarrow v})-\Psi_{\alpha}(M_t)\right) 
		+ \frac{f(j)}{F(t)} \frac{1}{d(v)}  \left(\Psi_{\alpha}\left( M_t\vert_{v \rightarrow u}\right)-\Psi_{\alpha}(M_t)\right)\\
  =&\sum_{\substack{\{u,v\} \in E\\ u \in V_\alpha(t)\\ v \in V_j(t), j \neq \alpha}}  \frac{f^+}{F(t)} \frac{1}{d(u)}  \frac{1}{d(v)} 
		- \frac{f(j)}{F(t)} \frac{1}{d(v)}  \frac{1}{d(u)}
\end{align*}  
Next observe that $F(t) < n  f^+$ and $f(j)\leq f^\ast$ for all $j\in\tau\setminus\alpha$. Thus 
\[ \E(\Psi_{\alpha}(M_{t+1}) - \Psi_{\alpha}(M_t)) > \frac{f^+ - f^\ast}{n f^+} \sum_{\substack{\{u,v\} \in E\\ u \in V_\alpha(t)\\ v \in V_j(t), j \neq \alpha}} \frac{1}{d(u)d(v)} > \left(1- \frac{f^\ast}{f^+} \right)/n^{3},  \]
concluding the proof.
\end{proof}

The technique that we use to upper bound $\E(\absorb_\alpha(G,\tau,f,M_0))$ relies on 
martingale techniques which are made explicit by Hajek~\cite{hajek-HittingTimeOccupationTimeBounds-1982}. 
Diaz et.~al \cite[Theorem 6]{diaz-ApproximatingFixationProbabilities-2014} used this technique to study the $2$-type process, following the approach of He and Yao~\cite{he-DriftAnalysisAverage-2001}.
For convenience, we use here a multi-type generalisation of the statement from Goldberg et al.~\cite{anngoldberg-PhaseTransitionsMoran-2020}.
\begin{lemma}[\cite{anngoldberg-PhaseTransitionsMoran-2020} Lemma 45]\label{lem:drift_lemma}
    Let $Y$ be a Markov Chain with a finite state space $\Omega_Y$. Let $c_1,c_2 >0$, let $\Psi: \Omega_Y \to \mathbb{R}_{\geq 0}$ be a function, and let $\mathsf{A} \geq 0$ be a stopping time with 
$\mathsf{A}\leq \min\{t\mid \mbox{$\Psi(Y_t)=0$ or $\Psi(Y_t)\geq c_1$} \}$. Suppose that
    \begin{enumerate}
        \item[(i)] from every state $S_1\in \Omega_Y$ with $0<\Psi(S_1)<c_1$, there is a path in $Y$ from $S_1$ to some state~$S_2$ with $\Psi(S_2)=0$ or $\Psi(S_2)\geq c_1$;
        \item[(ii)] for all $t\geq 0$, if $\Psi(Y_t)<c_1$ then $\Psi(Y_{t+1})\leq c_1 + 1$; and
        \item[(iii)] for all $t\geq 0$ and all $S\in \Omega_Y$ such that the events $\mathsf{A}>t$ and $Y_t=S$ are consistent,   $\E(\Psi(Y_{t+1}) - \Psi(Y_t)\mid Y_t=S)\geq c_2$.
    \end{enumerate}
    Then we have $\E(\mathsf{A})\leq (c_1 - \Psi(Y_0)+1)/c_2$.
\end{lemma}

\begin{corollary}\label{cor:tau-k-exp-bound}
Let $G=(V,E)$ be a connected graph, let $\tau$ be a set of types with $|\tau|=k>1$, and let $f\colon \tau \to \mathbb{Q}_{\geq 1}$ be a fitness function. Let $\alpha$
be a type in $\taumax{f}$. 
Let $f^* =\max\{f(j) \mid j\in\tau\setminus\alpha\}$, and suppose that $f^*<f(\alpha)=f^+$. 
Consider a state $M_0\in \Omegainit(G,\tau)$. Then \begin{align*}
	 	\E(\absorb_\alpha(G,\tau,f,M_0)) \leq \frac{f^+}{f^+-f^\ast} \left(n+1 \right) n^3 .
	 \end{align*}
\end{corollary}
\begin{proof}
Set $c_1:= \sum_{v \in V} 1/d(v)$ and  
$c_2:=( 1-{f^\ast}/{f^+})/ n^{3}$.
Let $\absorb_\alpha =\absorb_\alpha(G,\tau,f,M_0)$ and let $\Psi_\alpha$ be the potential function from  Definition~\ref{def:psi}.  
Let $M= M(G,\tau,f,M_0)$.
By construction
 $
	\absorb_\alpha = \min \{t\mid \mbox{$\Psi(M_t)=0$ or   $\Psi(M_t)=c_1$}\}. 
 $
If $V_\alpha(0) \in \{V,\emptyset\}$  then $\absorb_\alpha = 0$ and we are done. 
Otherwise, note that the state space~$\Omega$ of $M$ is finite and that $\Psi_\alpha$ is non-negative for every state. We now show that all three requirements $(i)-(iii)$ of Lemma~\ref{lem:drift_lemma} are satisfied. 

For $(i)$, suppose $0<\Psi(M_t)<c_1$ for some non-negative integer~$t$.  
From the definition of~$A_\alpha$, $V_\alpha \not\in \{\emptyset,V\}$. Since $G$ is connected, 
there is a positive probability that $M$ moves from $M_t$ to a state~$S$ in which $V_\alpha = V$ so $\Psi(S)=c_1$. For $(ii)$,  note that the maximum potential is $c_1$.
Finally, for $(iii)$,   Lemma~\ref{lem:potential-increasing-in-expectation_new} ensures that $\E(\Psi_\alpha(M_{t+1})-\Psi_\alpha(M_t)) > c_2>0$.  

Lemma~\ref{lem:drift_lemma} thus
implies
$\E(A_\alpha)\leq (c_1 - \Psi_\alpha(M_0)+1)/c_2 \leq (c_1 +1)/c_2 $.
Since $c_1\leq n$, the corollary follows. 
\end{proof}

\subsection{Expected total absorption time} \label{sec:absorption:weaker}

Now that we have an upper bound on the expectation of $A_\alpha(G,\tau,f,M_0)$
for the most-fit type~$\alpha$
we use it to derive  an upper bound on the expectation of~$A_j(G,\tau,f,M_0)$ for any 
$j\in \tau$. Our strategy is as follows.
If $j=\alpha$, we are finished. Otherwise after time $A_\alpha(G,\tau,f,M_0)$
there are two possibilities. If type~$\alpha$ has fixated, then $A_j(G,\tau,f,M_0)$ occurs.
Otherwise, type~$\alpha$ has become extinct, so we repeat with the new set of types $\tau \setminus \alpha$.
 This inductive approach applies to the setting where the fitnesses of the types
 are pairwise distinct, which we will assume in the first part of this section. The general result (without this assumption) will be derived afterwards. 

When we assume that the fitnesses are distinct we will typically assume that $\tau = \{1,\ldots,k\}$ and that $f(1)<f(2)<\dots <f(k)$. This simplifies the notation with no loss of generality. Note that $\taumax{f} = \{k\}$ in this case.

\begin{definition}[$\absorb_{\geq j}$] \label{def:tau-j-ind}
Let $G$ be a connected graph. Let $k$ be a positive integer.
Let $\tau=\{1,\ldots,k\}$. Let
$f\colon \tau \to \mathbb{Q}_{\geq 1}$ be a fitness function such that
$f(1)<f(2)<\dots <f(k)$. 
Let $M_0$ be a state in~$\Omegainit(G,\tau)$ and let $M = M(G,\tau,f,M_0)$.
For any type $j \in \{2,\dots,k\}$, define
\begin{align*}
	\absorb_{\geq j} :=  \min \{ t \in \nonnegints\mid
 \mbox{
 $\cup_{i \geq j} V_i(t) = \emptyset$  or, for some $i \geq j$,  $V_i(t) = V$ 
 }
 \}.
\end{align*}
\end{definition}
$\absorb_{\geq j}$ is thus the first point in time when either all types $j , \ldots, k$ are extinct or one of them has fixated.

\begin{observation}\label{lem:tau-j-ind-monotone}
Let $G$ be a connected graph. Let $k$ be a positive integer.
Let $\tau=\{1,\ldots,k\}$. Let
$f\colon \tau \to \mathbb{Q}_{\geq 1}$ be a fitness function such that
$f(1)<f(2)<\dots <f(k)$. 
Let $M_0$ be a state in~$\Omegainit(G,\tau)$ and let $M= M(G,\tau,f,M_0)$.
Then
$\absorb_{\geq j}$ is monotonically decreasing in $j$. That is, for all $j \in \{2,\dots,k\}$,
$\absorb_{\geq j} \geq \absorb_{\geq j+1}$.
\end{observation}

\begin{proof}
At time $t=\absorb_{\geq j}$, either 
$\cup_{i \geq j} V_i(t) = \emptyset$  or, for some $i \geq j$,  $V_i(t) = V$.
The only case that doesn't directly imply
  $\absorb_{\geq j+1} \leq t$ (according to the definition)
  is $V_j(t)=V$ but this case implies
$\cup_{i\geq j+1} V_i(t) = \emptyset$, hence 
    $\absorb_{\geq j+1} \leq t$.  
\end{proof}

\begin{lemma} \label{lem:exp-tau-j-bound}
Let $G$ be a connected graph with $n$ vertices. Let $k$ be a positive integer.
Let $\tau=\{1,\ldots,k\}$. Let
$f\colon \tau \to \mathbb{Q}_{\geq 1}$ be a fitness function such that
$f(1)<f(2)<\dots <f(k)$. 
Let $M_0$ be a state in~$\Omegainit(G,\tau)$ and let $M= M(G,\tau,f,M_0)$.
Let $j$ be a type in $\{2,\ldots,k\}$. Then 
	\begin{align*}
		 \E(\absorb_{\geq j}) \leq \sum_{i=j}^{k} \frac{f(i)}{f(i)-f(i-1)}(n+1)n^3.
	\end{align*}
\end{lemma}

\begin{proof}
The lemma is proved by induction on $k-j$.
For the base case $j=k$ note that $\absorb_{\geq k}= {\absorb_k}$. Thus the claim holds 
for $j=k$ by Corollary~\ref{cor:tau-k-exp-bound}. 

For the inductive step, assume that the claim holds for $j+1$; we will show that it holds for $j$.

To this end, let $\Omega^j$ be the set of all states 
with no vertices assigned to types $\{j+1,\ldots,k\}$. Using Corollary~\ref{cor:tau-k-exp-bound} on initial states in $\Omega^j$, we obtain
\begin{align*}
			\E(\absorb_{\geq j}-\absorb_{\geq j+1})
			\leq& \sum_{S \in \Omega^j} \E(\absorb_{\geq j}-\absorb_{\geq j+1} \mid M_{\absorb_{\geq j+1}} = S)\cdot \Pr(M_{\absorb_{\geq j+1}} = S) \\
			\leq &\frac{f(j)}{f(j)-f(j-1)} (n+1) n^3 \cdot \underbrace{\sum_{S \in \Omega^j} \Pr(M_{\absorb_{\geq j+1}} = S)}_{\leq 1}.
		\end{align*}
Consequently, in combination with the induction hypothesis for $j+1$, we have
		\begin{align*}
			\E(\absorb_{\geq j})
			&= \E(\absorb_{\geq j}-\absorb_{\geq j+1})+\E(\absorb_{\geq j+1}) \\
			&\leq \frac{f(j)}{f(j) -f(j-1)} (n+1)n^3+\E(\absorb_{\geq j+1})  \\
			&\leq  \frac{f(j)}{f(j) -f(j-1)} (n+1)n^3+ \sum_{i=j+1}^{k}  \frac{f(i)}{f(i) -f(i-1)} (n+1)n^3, \\ 
		\end{align*}
		concluding the proof.
\end{proof}

In the final step, we extend our bounds on the absorption times to processes in which distinct types may have the same fitness.

\begin{theorem}\label{thm:absorption_time}
Let $G=(V,E)$ be a connected graph with $n$ vertices, let $\tau$ be a set of types with $|\tau|>1$, let 
$f\colon \tau \to \mathbb{Q}_{\geq 1}$ be a fitness function, and let $M_0$ be a state in~$\Omegainit(G,\tau)$. 
Suppose that the types in~$\tau$ have $k$~distinct fitnesses and let $f^-=f_1<f_2<\dots<f_k=f^+$ be these fitnesses. 
For each $i\in \{1,\ldots,k\}$, let
$\ell_i = |f^{-1}(f_i)|$ be the number of types with fitness $f_i$  and let $\ell = \max\{ \ell_i\mid i\in \{1,\dots,k\} \}$ be the maximum number of types that share the same fitness.
 Then
\begin{itemize}
    \item[(1)] The total absorption time is bounded by  
    \[\E(\absorb(G,\tau,f,M_0)) \leq (\ell-1)n^6 + \sum_{i=2}^k \frac{f_i}{f_i-f_{i-1}} (n+1)n^3.\]
    \item[(2)] For any $j\in \{2,\dots,k\}$ the absorption time of each type $\beta$ with $f(\beta)=f_j$ is bounded by
    \[\E(\absorb_\beta(G,\tau,f,M_0)) \leq (\ell_j-1)n^6 + \sum_{i=j}^k \frac{f_i}{f_i-f_{i-1}} (n+1)n^3.\]
    \item[(3)] The absorption time of each type $\beta$ with $f(\beta)=f_1$ is bounded by 
    \[\E(\absorb_\beta(G,\tau,f,M_0)) \leq (\ell_1-1)n^6 + \sum_{i=2}^k \frac{f_i}{f_i-f_{i-1}} (n+1)n^3.\]
\end{itemize}
\end{theorem}
\begin{proof}
We first treat types with equal fitnesses as the same type. 
Let $\tau'=\{1,\dots,k\}$ and let 
$f'\colon \tau' \to \mathbb{Q}_{\geq 1}$
  be the fitness function such that, for each $i\in \tau'$, $f'(i) = f_i$. 
Let $\Omega$ be the set of functions from~$V$ to~$\tau$ and
let $\Omega'$ be the set of functions from~$V$ to~$\tau'$. 
As usual, let $\Omegainit(G,\tau)$ be the set of functions in~$\Omega$ with range~$\tau$ and let
$\Omegainit'(G,\tau')$ be the set of functions in~$\Omega'$  with range~$\tau'$.

We will define a function~$g$ from~$\Omega$ to~$\Omega'$.
Given $S\in \Omega$, let $g(S)$ be the state in~$\Omega'$ defined as follows.
For each $v\in V$, identify the value $i\in \{1,\ldots,k\}$ such that $f(v) = f_i$.
Then set $g(S)(v) = i$.
Note that if $S\in \Omegainit(G,\tau)$ then $S' \in \Omegainit'(G,\tau')$.
Let $M'_0 = g(M_0)$.

Let $M = M(G,\tau,f,M_0)$ and $M' = M'(G,\tau',f',M'_0)$.
For any $j\in \tau'$, let  
$V'_j(t) = \{v \in V \mid M'_t(v) = j\}$.
Following Definition~\ref{def:abtime} and~\ref{def:tau-j-ind},
define 
\begin{align*}
\absorb'_j &=  \min \{ 
\mbox{$t \in \nonnegints \mid V'_j(t) = V $ or  $V'_j(t) = \emptyset$} \},\\
\absorb' &=   \min \{ t \in \nonnegints \mid 
\mbox{ for some    
    $j\in\tau$, $V'_j(t) = V$} 
    \}, \text{and}\\
	\absorb'_{\geq j} &=   \min \{ t \in \nonnegints\mid
 \mbox{
 $\cup_{i \geq j} V'_i(t) = \emptyset$  or, for some $i \geq j$,  $V'_i(t) = V$ 
 }
 \}.
\end{align*}

Since $\absorb'_{\geq 2} = \absorb'$, 
Lemma~\ref{lem:exp-tau-j-bound} applied to~$M'$ implies \begin{equation}\label{eq:Mprime_absorb}
        \E(\absorb') \leq \sum_{i=2}^k \frac{f_i}{f_i-f_{i-1}} (n+1)n^3 .
\end{equation}

Note that $g(M_0),g(M_1),\ldots$ is a faithful copy of the chain~$M'$.
If $M$ and $M'$ are coupled using this coupling, then at time~$\absorb'$, all vertices
are assigned types with the same fitness in  $M_{\absorb'}$.

Consider any state $S\in \Omega$ in which vertices are assigned $q$ different types, all with the same fitness.
Theorem~11 of the work of D{\'i}az et al.~\cite{diaz-ApproximatingFixationProbabilities-2014} shows that in the 2-type Moran process,
from any state with 2 types, the expected time until there is at most one type is at most $n^6$.
Applying this argument to~$S$ by distinguishing the first type from the $q-1$ others, the expected time until there are at most~$q-1$ types is at most~$n^6$.
Since there are at most~$\ell$ distinct types with any given fitness, 
$\E(\absorb(G,\tau,f,S)) \leq (\ell-1) n^6$.
The upper bound (1) on $\E(\absorb(G,\tau,f,M_0)) $ given in the statement of the theorem follows.

In order to prove (2), note that   $\absorb'_j\leq {\absorb'}_{\geq j}$. 
 Lemma~\ref{lem:exp-tau-j-bound} applied to~$M'$ implies  
 $$
        \E(\absorb'_{\geq j}) \leq \sum_{i=j}^k \frac{f_i}{f_i-f_{i-1}} (n+1)n^3.
$$
The rest of the argument is analogous to the argument for (1).

Finally, for the proof of (3) observe that after time $\absorb'$, either type $\beta$ is extinct, or all
remaining types have fitness $f(\beta)=f_1$. The argument is then the same as the other cases.  This concludes the proof.
\end{proof}

Bounding the absorption time from an initial state chosen from a distribution in $\calD(G,\tau)$ is an immediate consequence. 

\begin{definition}[Absorption time from random initial state] \label{def:abtime_random}
Let $G=(V,E)$ be a connected graph, let $\tau$ be a set of types, let $f:\tau \to \mathbb{Q}_{\geq 1}$ be a fitness function and let $D\in \calD(G,\tau)$.
For the Moran process $M = M(G, \tau, f,D)$ and for any type $j\in \tau$
we define the stopping time 
$$\absorb_j = \absorb_j(G,\tau,f,D) := \min \{ t \in \nonnegints\mid 
\mbox{$V_j(t) = V$ or  $V_j(t) = \emptyset$} \}.$$
We refer to $A_j$ as the absorption time of type~$j$ in~$M$ from~$D$. The total absorption time of $M$ from $D$
is defined as 
$$\absorb = \absorb(G,\tau, f,D) := \min \{ t \in \nonnegints\mid 
\mbox{for some $j\in \tau$, $ V_j(t) = V$ }\}.$$ 
\end{definition}

Since the support of any distribution $D\in \calD(G,\tau)$ 
is contained in $\Omegainit(G,\tau)$, the bounds in Theorem~\ref{thm:absorption_time} also hold if $M_0$ is chosen from $D$. We make this explicit for the most fit type, since this is the result that we will use to establish our FPTRAS.

\begin{corollary}\label{cor:absorption_distribution_max}
Let $G$ be a connected graph with $n$ vertices, let $\tau$ be a set of types with $|\tau|>1$, let 
$f\colon \tau \to \mathbb{Q}_{\geq 1}$
  be a fitness function, and let $D\in \calD(G,\tau)$. Let  $f^\ast =\max\{f(j)\mid j\in\tau\setminus\taumax{f}\}$. The total absorption time of any $\alpha\in\taumax{f}$ is bounded by
 \[\E(\absorb_\alpha(G,\tau,f,D)) \leq (|\taumax{f}|-1)n^6 + \frac{f(\alpha)}{f(\alpha)-f^\ast} (n+1)n^3.\]
\end{corollary}
\begin{proof}
 Note that
$\E(\absorb_\alpha(G,\tau,f,D)) = \sum_{S\in \Omega} \Pr_D(S) \, E(\absorb_\alpha(G,\tau,f,S))$.
Since the support of~$D$ is contained in $\Omegainit(G,\tau$), the corollary follows from Theorem~\ref{thm:absorption_time}.
\end{proof}

\section{Construction of the FPTRAS} \label{sec:algo}

\begin{reptheorem}{thm:main}
   \stateThmMain
\end{reptheorem}
\begin{proof}
Consider first the following algorithm $\mathbb{A}'$ that  takes as input a connected graph~$G$, 
a set of types~$\tau$ with $|\tau|>1$, 
a fitness function $f\colon \tau \to \mathbb{Q}_{\geq 1}$, 
a type $\alpha\in\taumax{f}$, accuracy parameters $\varepsilon$ and $\delta'$,
and (via oracle access) a distribution   $D\in \calD(G,\tau)$.  
By Corollary~\ref{cor:lower_bound_probability},  $\pi_\alpha(G,\tau,f,D)\geq 1/n$.
    
Set $t=\lceil  3\ln(2/\delta') n/\varepsilon^2 \rceil$. 
The algorithm $\mathbb{A}'$ performs
$t$ independent simulations of $M(G,\tau,f,D)$, 
stopping the simulation if  either 
(i) all vertices have type $\alpha$, or (ii) type $\alpha$ is extinct. 
In each simulation, the initial state is sampled using the oracle for $D$.

Let  $f^\ast =\max\{f(j)\mid j\in\tau\setminus\taumax{f}\}$. 
Corollary~\ref{cor:absorption_distribution_max} 
implies that the expected time of each simulation is at most 
$\mathsf{poly}(n)\cdot|\tau|\cdot f(\alpha)/(f(\alpha)-f^\ast)$.
Let $w$ be the number of simulations in which type $\alpha$ fixates.  The output of algorithm~$\mathbb{A}'$ is~$w/t$.

For $i\in \{1,\dots,t\}$ let $X_i$ be the indicator for the event that $\alpha$ fixates in the $i$'th simulation. Since $\E(X_i) = \pi_\alpha(G,\tau,f,D)\geq 1/n $, it follows that $\E(w/t) = \E(\tfrac{1}{t}  \sum_{i=1}^t X_i)= \pi_\alpha(G,\tau,f,D)$. By a Chernoff bound (see e.g.\ the standard textbook of Mitzenmacher and Upfal~\cite[Theorem 11.1]{MitzenmacherU05}),
\[ \Pr(|w/t - \pi_\alpha(G,\tau,f,D)| < \varepsilon \pi_\alpha(G,\tau,f,D)) \geq 1-\delta'.\]

The expected total running time of $\mathbb{A}'$ is bounded from above by $t\cdot \mathsf{poly}(n)\cdot|\tau|\cdot f(\alpha)/(f(\alpha)-f^\ast)$.  
We next obtain an $\epsilon$-approximation algorithm $\mathbb{A}$ from $\mathbb{A}'$ using the following standard construction.
The algorithm $\mathbb{A}$ takes as input $G$, $\tau$, $f$, $\alpha$, $\varepsilon$, and
(via oracle access) the distribution~$D$.
The algorithm $\mathbb{A}$ simulates   $\mathbb{A}'$ 
with these inputs and with $\delta'=1/8$  for at most $8t\cdot \mathsf{poly}(n)\cdot|\tau|\cdot f(\alpha)/(f(\alpha)-f^\ast)$ steps. If $\mathbb{A}'$ terminates within this number of steps, 
then algorithm~$\mathbb{A}$ returns its output.
Otherwise, algorithm~$\mathbb{A}$ stops the simulation and returns~$0$. 

By Markov's inequality, the probability that $\mathbb{A}$ stops the simulation is at most~$1/8$. Thus, the probability that $\mathbb{A}$ returns an $\varepsilon$-approximation to $\pi_\alpha(G,\tau,f,D)$ is at least $1 - \delta' - 1/8 \geq 3/4$.

The desired FPTRAS takes the median of the outputs after repeating
algorithm $\mathbb{A}$ 
$O(\log(1/\delta))$ times.   
\end{proof}

\begin{reptheorem}{thm:maxfix-mut-has-fptras}
   \stateThmMaxfixMutHasFptras
\end{reptheorem}
\begin{proof}
The proof is the same as the proof of Theorem~\ref{thm:main} except that 
the algorithms can directly sample from~$\dmut(G,\tau,\tauord)$ in polynomial time, without using the oracle. 
\end{proof}

\section{Upper-bounding fixation probabilities in terms of the 2-type Moran process}

Corollary~\ref{cor:tech} gives a lower bound on the fixation probability of a type~$\alpha$ in the multi-type Moran process, based on the fixation probability of~$\alpha$ in a related two-type Moran process.
It is just as easy to get upper bounds with this idea, as the following corollary shows.

\begin{corollary}\label{cor:techrev}
Let $G=(V,E)$ be a connected graph, let $\tau$ be a set of types, and let $f\colon \tau \to \mathbb{Q}_{\geq 1}$ be a fitness function. Let $\alpha$ be a type in~$\tau$.
Let $\beta = \arg\min\{f(j) \mid j\in \tau\setminus\{\alpha\}\}$.
Let $\tau'=\{\alpha,\beta\}$.
Let $\Omega$ be the set of functions from $V$ to $\tau$ and let $\Omega'$ be the set of functions from $V$ to $\tau'$. Let $g\colon \Omega \to \Omega'$ be defined as follows.
For every state $S\in \Omega$, $g(S)$ is the state  
such that $g(S)(v) = \alpha$ if $S(v)=\alpha$ and $g(S)(v)=\beta$, otherwise.
Then, for every 
state $M_0 \in \Omega$, $\pi_\alpha(G,\tau,f,M_0)\leq \pi_\alpha(G,\tau',f,g(M_0))$.
 \end{corollary}
\begin{proof}
Let  $f'\colon \tau \to \mathbb{Q}_{\geq 1}$ 
be the fitness function with $f'(\alpha)=\alpha$ and, for any $\beta' \in \tau\setminus \{\alpha\}$,
$f'(\beta') = f(\beta)$.
Note that  
$\pi_\alpha(G,\tau',f,g(M_0)) = \pi_\alpha(G,\tau,f',M_0)$, so the goal is to
prove $\pi_\alpha(G,\tau,f,M_0)\leq \pi_\alpha(G,\tau,f',M_0)$.

Let $\tildeM = \tildeM(G,\tau,f,M_0)$ and let $\tildeM' = \tildeM'(G,\tau,f',M_0)$.
For all $j\in \tau$, let $\tildeV_j(t)=\{v \in V \mid \tildeM_t(v) = j\}$ and 
$\tildeV'_j(t)=\{v \in V \mid \tildeM'_t(v) = j\}$.
Applying Lemma~\ref{lem:coupling} (with the roles of $f$ and $f'$ reversed)
we find that 
for any  time $t \in \nonnegreals$, $\tildeV_{\alpha}'(t)\supseteq \tildeV_{\alpha}(t)$.
so 
$\tildepi_\alpha(G,\tau,f,M_0) \leq \tildepi_\alpha(G,\tau,f',M_0)$.
As in the proof of Corollary~\ref{cor:tech} the final step is to observe that
$\tildepi_\alpha(G,\tau,f,M_0) = \pi_\alpha(G,\tau,f,M_0)$
and
$\tildepi_\alpha(G,\tau,f',M_0) = \pi_\alpha(G,\tau,f',M_0)$.
  \end{proof}

We conclude the paper by using Corollaries~\ref{cor:tech} and~\ref{cor:techrev} 
to re-derive the recent bounds of Ferreira and Neves
on fixation probabilities of the 3-type Moran process in the very special case where $G$ is a complete graph. Their main result on the Moran process~\cite[Theorem~4]{ferreira-FixationProbabilitiesMoran-2020} 
 corresponds to  the special case of Lemma~\ref{lem:bounds_Kn} with $\tau=\{A,B,C\}$ and $f(A)>f(B)>f(C)$.

\begin{lemma}\label{lem:bounds_Kn}\label{lem:complete}
Let $n$ be a positive integer and let $G=(V,E)$ be a complete graph with $n$ vertices.
Let $\tau$ be a set of types and let $f\colon \tau \to \mathbb{Q}_{\geq 1}$ be a fitness function.
Let $\alpha$ be a type in~$\tau$.  Let $\beta^+ = \arg \max \{f(\beta) \mid \beta\in \tau \setminus\{\alpha\}\}$ and
$\beta^-=\arg \min \{f(\beta) \mid \beta\in \tau \setminus\{\alpha\}\}$. 
Suppose that $f(\alpha) \neq f(\beta^+)$ and $f(\alpha) \neq f(\beta^-)$.
Let $M_0$ be a state in the set $\Omega$ of functions from~$V$ to~$\tau$.
Let $i$ be the number of vertices in~$V$
which are mapped to type~$\alpha$ by~$M_0$. Then 
\[ \frac{1-\left(\frac{f(\beta^+) }{ f(\alpha)}\right)^{i}\strut}{\strut1-\left(\frac{f(\beta^+)}{ f(\alpha)}\right)^n} \leq \pi_\alpha(G,\tau,f,M_0) \leq \frac{1-\left(\frac{f(\beta^-)}{ f(\alpha)}\right)^{i}}{\strut 1-\left(\frac{f(\beta^-)}{ f(\alpha)}\right)^n} .\]
\end{lemma}

\begin{proof}
We first give the lower bound.
Let $\beta = \beta^+$.
Let $\tau'=\{\alpha,\beta\}$.
Let   $\Omega'$ be the set of functions from $V$ to $\tau'$. Let $g\colon \Omega \to \Omega'$ be defined as follows.
For every state $S\in \Omega$, $g(S)$ is the state  
such that $g(S)(v) = \alpha$ if $S(v)=\alpha$ and $g(S)(v)=\beta$, otherwise.
Corollary~\ref{cor:tech} shows that  $\pi_\alpha(G,\tau,f,M_0)\geq \pi_\alpha(G,\tau',f,g(M_0))$.
To conclude the proof, we note that $\pi_\alpha(G,\tau',f,g(M_0))$ is 
equal to the lower bound given in the lemma statement. This well-known fact follows because $\pi_\alpha(G,\tau',f,g(M_0))$  is
the probability of
hitting the lower boundary in a biased one-dimensional random walk. 
The walk is biased since $f(\alpha) \neq f(\beta)$. See Lieberman et al.~\cite{lieberman-EvolutionaryDynamicsGraphs-2005} for details. 
The upper bound is the same, except $\beta = \beta^{-}$ and Corollary~\ref{cor:techrev} is
used instead of Corollary~\ref{cor:tech}.  \end{proof}

\addcontentsline{toc}{section}{References}
\bibliographystyle{plainurl}
\bibliography{references}

\begin{thebibliography}{10}

\bibitem{ArvindR02}
Vikraman Arvind and Venkatesh Raman.
\newblock Approximation algorithms for some parameterized counting problems.
\newblock In Prosenjit Bose and Pat Morin, editors, {\em Algorithms and
  Computation, 13th International Symposium, {ISAAC} 2002 Vancouver, BC,
  Canada, November 21-23, 2002, Proceedings}, volume 2518 of {\em Lecture Notes
  in Computer Science}, pages 453--464. Springer, 2002.
\newblock \href {https://doi.org/10.1007/3-540-36136-7\_40}
  {\path{doi:10.1007/3-540-36136-7\_40}}.

\bibitem{beerenwinkel-GeneticProgressionWaiting-2007}
Niko Beerenwinkel, Tibor Antal, David Dingli, Arne Traulsen, Kenneth~W.
  Kinzler, Victor~E. Velculescu, Bert Vogelstein, and Martin~A. Nowak.
\newblock Genetic {{Progression}} and the {{Waiting Time}} to {{Cancer}}.
\newblock {\em PLOS Computational Biology}, 3(11):e225, November 2007.
\newblock \href {https://doi.org/10.1371/journal.pcbi.0030225}
  {\path{doi:10.1371/journal.pcbi.0030225}}.

\bibitem{chao-SpatialMutationModel-2021}
Brian Chao and Jason Schweinsberg.
\newblock A {{Spatial Mutation Model}} with {{Increasing Mutation Rates}}.
\newblock {\em CoRR}, abs/2108.09590, 2021.
\newblock \href {http://arxiv.org/abs/2108.09590} {\path{arXiv:2108.09590}},
  \href {https://doi.org/10.48550/arXiv.2108.09590}
  {\path{doi:10.48550/arXiv.2108.09590}}.

\bibitem{chatterjee-FasterMontecarloAlgorithms-2017}
Krishnendu Chatterjee, Rasmus {Ibsen-Jensen}, and Martin~A. Nowak.
\newblock Faster monte-carlo algorithms for fixation probability of the moran
  process on undirected graphs.
\newblock In {\em 42nd International Symposium on Mathematical Foundations of
  Computer Science}, volume~83 of {\em {{LIPIcs}}}, pages 61:1--61:13. {Schloss
  Dagstuhl - Leibniz-Zentrum f\"ur Informatik}, 2017.
\newblock \href {https://doi.org/10.4230/LIPIcs.MFCS.2017.61}
  {\path{doi:10.4230/LIPIcs.MFCS.2017.61}}.

\bibitem{diaz-ApproximatingFixationProbabilities-2014}
Josep D{\'i}az, Leslie~Ann Goldberg, George~B. Mertzios, David Richerby, Maria
  Serna, and Paul~G. Spirakis.
\newblock Approximating {{Fixation Probabilities}} in the {{Generalized Moran
  Process}}.
\newblock {\em Algorithmica}, 69(1):78--91, May 2014.
\newblock \href {https://doi.org/10.1007/s00453-012-9722-7}
  {\path{doi:10.1007/s00453-012-9722-7}}.

\bibitem{diaz-AbsorptionTimeMoran-2016}
Josep D{\'i}az, Leslie~Ann Goldberg, David Richerby, and Maria Serna.
\newblock Absorption time of the {{Moran}} process.
\newblock {\em Random Structures \& Algorithms}, 49(1):137--159, 2016.
\newblock \href {https://doi.org/10.1002/rsa.20617}
  {\path{doi:10.1002/rsa.20617}}.

\bibitem{etheridge-CoalescentDualProcess-2009}
A.~M. Etheridge and R.~C. Griffiths.
\newblock A coalescent dual process in a {{Moran}} model with genic selection.
\newblock {\em Theoretical Population Biology}, 75(4):320--330, 2009.
\newblock \href {https://doi.org/10.1016/j.tpb.2009.03.004}
  {\path{doi:10.1016/j.tpb.2009.03.004}}.

\bibitem{ferreira-FixationProbabilitiesMoran-2020}
Eliza~M. Ferreira and Armando G.~M. Neves.
\newblock Fixation probabilities for the {{Moran}} process with three or more
  strategies: General and coupling results.
\newblock {\em Journal of Mathematical Biology}, 81(1):277--314, July 2020.
\newblock \href {https://doi.org/10.1007/s00285-020-01510-0}
  {\path{doi:10.1007/s00285-020-01510-0}}.

\bibitem{fisher-DominanceRatio-1923}
R.~A. Fisher.
\newblock On the {{Dominance Ratio}}.
\newblock {\em Proceedings of the Royal Society of Edinburgh}, 42:321--341,
  1923.
\newblock \href {https://doi.org/10.1017/S0370164600023993}
  {\path{doi:10.1017/S0370164600023993}}.

\bibitem{anngoldberg-PhaseTransitionsMoran-2020}
Leslie~Ann Goldberg, John Lapinskas, and David Richerby.
\newblock Phase transitions of the {{Moran}} process and algorithmic
  consequences.
\newblock {\em Random Structures \& Algorithms}, 56(3):597--647, 2020.
\newblock \href {https://doi.org/10.1002/rsa.20890}
  {\path{doi:10.1002/rsa.20890}}.

\bibitem{hajek-HittingTimeOccupationTimeBounds-1982}
Bruce Hajek.
\newblock Hitting-{{Time}} and {{Occupation-Time Bounds Implied}} by {{Drift
  Analysis}} with {{Applications}}.
\newblock {\em Advances in Applied Probability}, 14(3):502--525, 1982.
\newblock \href {https://doi.org/10.2307/1426671} {\path{doi:10.2307/1426671}}.

\bibitem{he-DriftAnalysisAverage-2001}
Jun He and Xin Yao.
\newblock Drift analysis and average time complexity of evolutionary
  algorithms.
\newblock {\em Artificial Intelligence}, 127(1):57--85, 2001.
\newblock \href {https://doi.org/10.1016/S0004-3702(01)00058-3}
  {\path{doi:10.1016/S0004-3702(01)00058-3}}.

\bibitem{lanchier-WrightFisherMoran-2017}
Nicolas Lanchier.
\newblock Wright–{{Fisher}} and {{Moran}} models.
\newblock In Nicolas Lanchier, editor, {\em Stochastic {{Modeling}}},
  Universitext, pages 203--218. {Springer International Publishing}, 2017.
\newblock \href {https://doi.org/10.1007/978-3-319-50038-6_12}
  {\path{doi:10.1007/978-3-319-50038-6_12}}.

\bibitem{lieberman-EvolutionaryDynamicsGraphs-2005}
Erez Lieberman, Christoph Hauert, and Martin~A. Nowak.
\newblock Evolutionary dynamics on graphs.
\newblock {\em Nature}, 433(7023):312--316, January 2005.
\newblock \href {https://doi.org/10.1038/nature03204}
  {\path{doi:10.1038/nature03204}}.

\bibitem{MitzenmacherU05}
Michael Mitzenmacher and Eli Upfal.
\newblock {\em Probability and Computing: Randomized Algorithms and
  Probabilistic Analysis}.
\newblock Cambridge University Press, 2005.
\newblock \href {https://doi.org/10.1017/CBO9780511813603}
  {\path{doi:10.1017/CBO9780511813603}}.

\bibitem{Moran58}
P.~A.~P. Moran.
\newblock Random processes in genetics.
\newblock {\em Mathematical Proceedings of the Cambridge Philosophical
  Society}, 54(1):60–71, 1958.
\newblock \href {https://doi.org/10.1017/S0305004100033193}
  {\path{doi:10.1017/S0305004100033193}}.

\bibitem{Nowak06}
Martin~A. Nowak.
\newblock {\em Evolutionary Dynamics: Exploring the Equations of Life}.
\newblock Harvard University Press, 2006.
\newblock \href {https://doi.org/10.2307/j.ctvjghw98}
  {\path{doi:10.2307/j.ctvjghw98}}.

\bibitem{vogelstein-CancerGenomeLandscapes-2013}
Bert Vogelstein, Nickolas Papadopoulos, Victor~E. Velculescu, Shibin Zhou,
  Luis~A. Diaz, and Kenneth~W. Kinzler.
\newblock Cancer {{Genome Landscapes}}.
\newblock {\em Science (New York, N.Y.)}, 339(6127):1546--1558, March 2013.
\newblock \href {https://doi.org/10.1126/science.1235122}
  {\path{doi:10.1126/science.1235122}}.

\bibitem{wright-EvolutionMendelianPopulations-1931}
Sewall Wright.
\newblock Evolution in {{Mendelian Populations}}.
\newblock {\em Genetics}, 16(2):97--159, 1931.
\newblock \href {https://doi.org/10.1093/genetics/16.2.97}
  {\path{doi:10.1093/genetics/16.2.97}}.

\end{thebibliography}

\end{document}